\RequirePackage{amsmath}
\documentclass[11pt]{article}
\usepackage[margin=1.2in]{geometry}
\usepackage{authblk}
\usepackage{amsmath, systeme, environ,amsthm}
\newtheorem{lemma}{Lemma}
\newtheorem{proposition}{Proposition}
\newtheorem{problem}{Problem}
\newtheorem{theorem}{Theorem}
\newtheorem{corollary}{Corollary}
\usepackage{graphicx} 
\usepackage{graphicx}
\usepackage{algorithm}
\usepackage[noend]{algpseudocode}
\usepackage{amsfonts} 
\usepackage{amssymb}
\usepackage{mathtools}
\usepackage{xspace}
\usepackage{enumerate}
\usepackage{cleveref}
\usepackage{float}
\usepackage{tabularx}
\usepackage{breqn}
\usepackage{multicol}
\usepackage{wrapfig}
\usepackage{stmaryrd}
\usepackage{macros}
\usepackage{tikz}
\usetikzlibrary{arrows,shapes,automata, calc, chains, matrix, trees, positioning, scopes}

\DeclareMathOperator{\num}{num}

\newcommand{\comment}[1]{}

\newcommand{\AR}[1]{\textcolor{blue}{AR: #1}}

\newcommand{\Aa}{\mathcal{A}}

\newcommand{\MM}{\mathcal{M}}
\newcommand{\UU}{\mathcal{U}}
\newcommand{\VV}{\mathcal{V}}

\newcommand{\NN}{\mathbb{N}}

\renewcommand{\norm}[1]{|#1|}

\begin{document}

\title{Reachability in Fixed {VASS}:\\ Expressiveness and Lower Bounds\thanks{
This
work is part of a project that has received funding from the
European Research Council (ERC) under the European Union’s Horizon
2020 research and innovation programme (Grant agreement No. 852769,
ARiAT).}}
%\titlerunning{}
\author[]{Andrei Draghici}
\author[]{Christoph Haase}
\author[]{Andrew Ryzhikov}
\affil[]{Department of Computer Science, University of Oxford, Oxford, United Kingdom}
\date{}
\maketitle
\begin{abstract}
  The recent years have seen remarkable progress in establishing the complexity of the reachability problem for vector addition systems with
  states (VASS), equivalently
  known as Petri nets. Existing work 
  primarily considers the case in which both
  the VASS as well as the initial and target
  configurations are part of the input. In this
  paper, we investigate the reachability problem
  in the setting where the VASS is fixed
  and only the initial configuration
  is variable. 
  We show that fixed VASS fully express
  arithmetic on initial segments of the
  natural numbers. It
  follows that there is a very weak reduction from
  any fixed such number-theoretic predicate 
  (e.g.\ primality or square-freeness) to reachability
  in fixed VASS where configurations are 
  presented in unary. 
  If configurations are given in binary, we
  show that there is a fixed VASS with five
  counters
  whose reachability problem is PSPACE-hard.
\end{abstract}

\section{Introduction}
Vector addition systems with states (VASS), equivalently known as Petri nets, are a fundamental model of computation. A VASS comprises a finite-state controller
with a finite number of counters ranging over the non-negative integers. When a transition is
taken, counters can be updated by adding an integer, provided that the resulting 
counter values are all non-negative; otherwise the transition blocks. Given two configurations
of a VASS, each consisting of a control 
state and an assignment of values to the counters, the reachability problem asks whether
there is a path in the infinite transition system connecting the two configurations.
The VASS
reachability problem has been one of the most intriguing problems in
theoretical computer science and studied for more than fifty years.
In the 1970s, Lipton showed this problem EXPSPACE-hard~\cite{Lipton1976}.
Ever since the
1980s~\cite{Mayr84,Kosaraju82,LS15}, the reachability problem
has been known to be decidable, albeit with non-elementary complexity.
This wide gap between the EXPSPACE lower bound and a non-elementary
upper bound persisted for many years, until a recent series of papers
established various non-elementary lower bounds~\cite{czerwinski2020reachability,CO21,Leroux21}, and resulted
in showing the VASS reachability problem Ackermann-complete. The
lower bounds for this result require an unbounded number of counters, but even for a
fixed number of counters, the Petri net reachability problem requires
non-elementary time~\cite{CO21,Leroux21}. Further lower bounds for resource-bounded
or structurally restricted VASS have also been established recently~\cite{CO22}.

\paragraph{Main results.}
The main focus of this paper is to investigate the reachability problem for 
\emph{fixed VASS},
where the VASS under consideration is fixed and only the initial configuration forms the input to a reachability query. Here, it is crucical
to distinguish between the encoding of numbers used to represent counter values
in configurations: in \emph{unary encoding}, the representation length of a
natural number $n\in \N$ is its magnitude $n$ whereas in \emph{binary encoding}
the bit length of $n\in \N$ is $\lceil \log n\rceil + 1$. 
It turns out that establishing meaningful lower bounds under unary encoding 
of configurations
is a rather delicate issue; a full discussion is deferred to \Cref{sec-unaryVASS}.
As a first step, we establish a tight correspondence between reachability
in VASS and the first-order theory of initial segments of 
$\N$ with the arithmetical relations
addition ($+$) and multiplication ($\times$). An initial segment in $\N$ is
a set $\underline{N}=\{0,\ldots,N\}$ for some arbitrary but fixed $N\in \N \setminus \{ 0 \}$. Relations definable in this family of structures are known as \emph{rudimentary 
relations} and contain many important number-theoretic relations, c.f.~\cite{EM98}
and the references therein. For instance, the fixed formula $\mathrm{PRIMES}(x) 
\equiv \neg(x=0) \wedge \neg(x=1) \wedge \forall y<x\, \forall z<x\, \neg (x = y \times z)$
evaluates to true in $\underline N$ for all prime numbers up to $N$. 

We show
that for any fixed rudimentary relation $\Phi(x_1,\ldots, x_k)$, there is a fixed VASS
$V$ with a number of counters linear in the size of $\Phi$ such that $\Phi(n_1,\ldots,n_k)$
evaluates to true in $\underline N$ if and only if $V$ reaches a final control state
with all counters zero starting in an initial configuration in which the value
of the counter $c_i$ is obtained by application of some fixed polynomial $p_i$ to $n_1,\ldots,
n_k$ and $N$. It thus follows that reachability in fixed VASS under unary encoding
of configurations is hard for evaluating any rudimentary relation under unary 
encoding of numbers, and hence fixed VASS can, e.g. determine primality and 
square-freeness of a number given in unary. From those developments, it is already
possible to infer that reachability in fixed VASS with configurations encoded 
in binary is hard for every level of the polynomial hierarchy by a reduction
from the validity problem for short Presburger arithmetic~\cite{NguyenP22}. 
In fact, we can establish a PSPACE lower bound for reachability in VASS with
configurations encoded in binary by a generic reduction allowing to simulate
space-bounded computations of arbitrary Turing machines, encoded as natural
numbers, via a fixed VASS with only fice counters. A recent conjecture
of Jecker~\cite{Jecker23} states, for every VASS $V$, there exists a fixed constant
$C$ such that if a target configuration is reachable from an initial configuration
then there exists a witnessing path whose length is bounded by $C\cdot m$, where
$m$ is the maximum constant appearing in the initial and final configurations.
Thus, assuming Jecker's conjecture, reachability in fixed VASS under binary 
encoding of confiugrations would be PSPACE-complete.

\paragraph{Related work.}
To the best of our knowledge, the reachability problem for fixed VASS 
has not yet been systematically explored. Closest to the topics of this
paper the work by Rosier and Yen~\cite{RY86}, who conducted a multi-parameter
analysis of the complexity of the boundedness and coverability problems
for VASS.

However, the study of the computation power
of other fixed machines has a long history in the theory of computation. The
two classical decision problems for machines of a computation model are \emph{membership} (also called the \emph{word problem}) and \emph{reachability}. Membership asks whether a given machine accepts a given input; the (generic) reachability problem asks
whether given an initial and a target configuration,
 there is a path in the transition system induced by a given machine from the
 initial configuration to the target configuration.
The most prominent example of a reachability problem is the halting problem for different kinds of machines. Classically, the computational complexity of such problems assumes that both the computational model and its input word (for membership) or configurations (for reachability) are part of the input. However, these are two separate parameters. For example, in database theory, the database size and the query size are often considered separately, since the complexity of algorithms may depend very differently on these two parameters, and the sizes of these two parameters in applications can also vary a lot \cite{Vardi1982}. One approach to study such phenomena is to fix one of either the database or the query. More generally, the field of parameterised complexity studies the
computational difficulty of a problem with respect to multiple parameters
of the input.

%https://www.cs.rice.edu/~vardi/papers/sigmod08.pdf
Returning to our setting, this means fixing either the machine or its input. In this paper, we concentrate on the former. The question can then be seen as follows: in relation to a problem such as membership or reachability, which machine is the hardest one in the given computation model? For some models, the answer easily follows from the existence of universal machines, i.e., machines which are able to simulate any other machine from their class. A classical example here is a universal Turing machine. Sometimes the ability to simulate all other machines has to be relaxed, for example as for Greibach's hardest context-free language \cite{Greibach1973}. Greibach showed that there exists a fixed context-free grammar such that a membership query for any other context-free grammar can be efficiently reduced to a membership query for this grammar. Similar results are known for two-way non-deterministic pushdown languages~\cite{Rytter1981,Chistikov2022}.

\section{Preliminaries}
We denote by $\Z$ and $\N$ the set of integers and non-negative
integers, respectively. For $N\in\N$ we write $\underline{N}$ to denote the set $\{0,\dots,N\}$. By $[n, m]$ we define the set of integers between $n$ and $m$: $[n, m] = \{k \in \Z \mid n \le k \le m\}$.

\subsection{Counter automata}

A \emph{$d$-counter automaton} is a tuple $\Aa = (Q, \Delta, \zeta, q_0, q_f)$, where $Q$ 
is a finite set of states, $\Delta \subseteq Q \times \Z^d \times Q$ is the transition relation, 
$\zeta: \Delta \to \mathcal [1,d] \cup \{ \top \}$ is a function indicating which counter is tested
for zero along a transition ($\top$ meaning no counter is tested), $q_0 \in Q$ is the initial state, and $q_f\in Q$ is the 
final state. We assume that $q_f$ does not have outgoing transitions.

The set of configurations of $\Aa$ is $C(\Aa) \defeq \{ (q,n_1,\ldots,n_d) : q\in Q, n_i \in \N, 
1\le i\le n \}$. 
A run $\varrho$ of a counter automaton $\Aa$ from a configuration $c_1\in C(\Aa)$ to
$c_{n+1} \in C(\Aa)$ is a sequence of configurations interleaved with transitions
\[
\varrho = c_1 \xrightarrow[]{t_1} c_2 \xrightarrow[]{t_2} \dots \xrightarrow[]{t_n} c_{n+1}
\]
such that for all $1\le i\le n$, $c_i=(q,m_1,\dots,m_d)$ and $c_{i+1} = (r,m_1',\dots,m_d')$,
\begin{itemize}
\item $t_i = (q,(z_1,\ldots,z_d), r)$ with $m_j'=m_j + z_j$ for all $1\le j\le d$; and
\item $m_j = 0$ if $\zeta(t_i) = j$.
\end{itemize}

Observe that we can
without loss of generality
assume that each transition $t \in \Delta$ is of one of the two types:
\begin{itemize}
\item either no counter is tested for zero along $t$, that is, $\zeta(t) = \top$, in which case we call it \emph{an update transition};
\item or $t$ does not change the values of the counters, that is, $\zeta(t) = j$ and $t = (q, (0, \ldots, 0), r)$ in which case we call it \emph{a zero-test transition}.
\end{itemize}

We say that $\Aa$ is a \emph{vector addition system with states (VASS)} if $\Aa$ cannot
perform any zero tests, i.e., $\zeta$ is the constant function assigning $\top$ to 
all transitions. We can now formally define the main decision problem we study in this paper.

\begin{problem}\label{prob:acceptance-fixed-vass}
  \textsc{Fixed VASS Zero-Reachability} \\
  \textbf{Fixed:} $d$-VASS $\mathcal{A}$.\\
  \textbf{Input:} A source configuration $\vect x\in \N^d$. \\
  \textbf{Output:} Yes if and only if $\mathcal{A}$ has a run from $(q_0,\vect x)$ to $(q_f,\vect 0)$.
\end{problem}

\subsection{Counter programs}
\begin{wrapfigure}{r}{0.25\textwidth}
  \begin{algorithmic}[1]
    \State $\Gotoo{2}{4}$
    \State ~~ $x\minuseq 3$
    \State ~~ $\Goto 1$
    \State $x\pluseq 1$
    \State $\textbf{halt}$
  \end{algorithmic}
  \caption{Example of a counter program.}
  \label{fig:cp-example}
\end{wrapfigure}
For ease of presentation, we use the notion of counter programs
presented e.g. in \cite{czerwinski2020reachability}, which are equivalent to VASS,
and allow for presenting VASS (and counter automata) in a serialised way. A counter program is a primitive imperative program that executes arithmetic operations on a finite number of counter variables. Formally, a \emph{counter program} 
consists of a finite set $\mathcal X$ of global counter variables
(called \emph{counters} subsequently for brevity) ranging over the
natural numbers, and a finite sequence $1, \ldots, m$ of line numbers (subsequently
\emph{lines} for brevity), each associated with an instruction manipulating the values of the counters or a control flow
operation. Each instruction is of one the following forms:
\begin{itemize}
    \item $x\pluseq c$ (increment counter $x$ by constant $c \in \N$),
    \item $x\minuseq c$ (decrement counter $x$ by constant $c \in \N$),
    \item $\Gotoo{L_1}{L_2}$ (non-deterministically jump to the
      instruction labelled by $L_1$ or $L_2$),
    \item $\Skip$ (no operation).
\end{itemize}
We write $\Goto{L}$ as an abbreviation for $\Gotoo{L}{L}$, and also
allow statements of the form $\Goto{L_1}$ \textbf{or}
$L_2$ \textbf{or} $\cdots$ \textbf{or} $L_k$. Moreover, the line with the largest number is a special instruction \textbf{halt}.

An example of a counter program is given
in \Cref{fig:cp-example}. This counter program uses a single counter
$x$ and consists of five lines.
Starting in Line~1, the program non-deterministically loops and
decrements the counter $x$ by three every time, until it increments $x$ by one
and terminates. 

To be able to compose counter programs, we describe the operation of substitution, which substitutes a given line (which we always assume to have a \textbf{skip} instruction) of a counter program with the ``code'' of another counter program. Formally, let $C_1, C_2$ be counter programs with $m_1$ and $m_2$ lines respectively. The result of substituting Line $k$, $1 \le k \le m_1 - 1$, of $C_1$ with $C_2$ is a counter program $C'_1$ with $m_1 + m_2-1$ lines. The instruction corresponding to a line $L$, $1 \le L < m_1 + m_2$ are defined as follows: 
\begin{itemize}
    \item if $L < k$, it is the instruction of Line $L$ in $C_1$,

\item if $k \le L < m_2 + k$, it is the instruction of Line $L - k+1$ in $C_2$, 

\item if $L = m_2 + k-1$, it is the instruction \textbf{skip},

\item if $m_2 + k \le L$, it is the instruction of Line $L - k$ in $C_1$.
\end{itemize}

The line numbers in \textbf{goto} instructions are changed accordingly. We also consider a substitution of several counter programs. When specifying counter programs, to denote substitution of another counter program we just write its name instead of an instruction in a line. Also, we write $C_1;C_2$ for
\begin{algorithmic}[1]
  \State $C_1$
  \State $C_2$
\end{algorithmic}
and $C_1\textbf{ or }C_2$ as syntactic sugar for the counter program:
\begin{algorithmic}[1]
  \State $\Gotoo{2}{4}$
  \State $C_1$
  \State $\Goto 5$
  \State $C_2$
  \State $\Skip$
\end{algorithmic}

When $C$ is a counter program, we write
$\Looop$ $C$ as an abbreviation for the counter program
\begin{algorithmic}[1]
  \State $\Gotoo{2}{4}$
  \State $C$
  \State $\Goto 1$
  \State $\Skip$
\end{algorithmic}
Hence, the counter program in \Cref{fig:cp-example} corresponds to
$\Looop\ x \minuseq 3;\,x \pluseq 1$. We use indentation to mark the scope of the \textbf{loop} instruction.

\subsection{Runs of counter programs.}

Exactly as in the case of VASS, a \emph{configuration} of a counter program is an element $(L, f) \in \N \times \N^{\mathcal{X}}$, where $L\in \N$ is a program line with a
corresponding instruction, and $f\colon \mathcal X \to \N$ is a counter valuation. The semantics of counter programs are defined in a natural way: after executing the
instructions on the line $L$, we either non-deterministically go to one of the specified lines (if the instruction on line $L$ is a \textbf{goto} instruction), and otherwise, we go to the line $L + 1$. After executing the last line, we stop.

A \emph{run} of a counter program is a sequence $\varrho\colon (L_1,\f_1) \xrightarrow{} (L_2,\f_2) \xrightarrow{}
\cdots \xrightarrow{}(L_n,\f_n)$ of configurations defined naturally according to the described semantics. For example, $(1,\{x\mapsto 7\}) \xrightarrow{} (4,\{x\mapsto 7\}) \xrightarrow{} (5,\{x \mapsto 8\})$
is a run of the counter program in \Cref{fig:cp-example}. 

One can view a counter program as a VASS by treating line numbers as states and defining transitions as specified
by the counter program, each labelled with the respective instruction. It is also easy to see how to convert a VASS into a counter program.

Given a run
$\varrho\colon (L_1,\f_1) \xrightarrow{} (L_2,\f_2) \xrightarrow{}
\cdots \xrightarrow{} (L_n,\f_n)$, we say that $\varrho$
is \emph{zero-terminating} if $L_1 = 1$, the instruction on line $L_n$ is \textbf{halt}, and  $\f_n(x)=0$ for all $x\in\mathcal{X}$. We denote by $\val_{end}(\varrho,x)
\defeq \f_n(x)$
the value of the counter $x$ once the run terminates. Sometimes, we
also want to talk about the value of a counter at a specific point
during the execution of a run and define $val_i(\varrho,x)$ to be the value of the counter $x$ right before we execute the instruction on line $i$ in the run $\varrho$ for the first time, i.e.\ 
$\val_i(\varrho,x)\defeq f_k(x)$, where $k$ is the smallest index such that $L_k=i$. For instance, in the example above,
we have $\val_{end}(\varrho,x)=8$ and $\val_4(\varrho,x)=7$. We often construct
counter programs that admit exactly one run~$\varrho$ from a given initial
configuration to a target configuration. In such a setting, we may omit the reference to
$\varrho$ and simply write $\val_{end}(x)$ and $\val_i(x)$. The effect
$\eff(\varrho)\colon \mathcal X \to \Z$ of a run $\varrho$ starting in $(1,f_1)$ and ending in
$(n,f_n)$ is a map such that
$\eff(\varrho, x)=f_n(x)-f_1(x)$ for all $x\in \mathcal X$.

In the context of counter programs, the zero-reachability problem takes the following form.

\begin{problem}\label{prob:acceptance fixed C}
  \textsc{Fixed counter program zero-reachability} \\
  \textbf{Fixed:} Counter program $\mathcal{C}$.\\
  \textbf{Input:} An initial configuration $\vect x\in \N^m$. \\
  \textbf{Output:} Yes if and only if $\mathcal{C}$ has a zero-terminating run from $\vect x$.
\end{problem}

\section{Implementation of zero tests}\label{sec-gadgets}
The structure of runs in arbitrary counter programs is very complicated and hard to analyse, and hence it is difficult to force a counter program to have a prescribed behaviour. One of the common ways to deal with this issue is to introduce some restricted zero tests, that is, some gadgets that guarantee that if a run reaches a certain configuration, then along this run, the values of some counters are zero at prescribed positions. In this section, summarising~\cite{czerwinski2020reachability},
we describe such a gadget in the case where the values of counters are bounded by a given number. The number of zero tests that can be performed this way is also bounded. For a counter $v$, we call this gadget $\Zero{v}$, and later on we will use it as a single instruction to test that the value of $v$ is zero before executing it.

In Section \ref{sec-unaryVASS}, the assumption that the values of the counters are bounded comes from the the fact that the corresponding values of the variables in rudimentary arithmetic are bounded. In \Cref{sec-binaryVASS-new}, we enforce this property for more powerful models
of computation and show how to simulate these models with VASS.

Let $N\in\N$ be an upper bound on the value of a counter $v$. Then, we can introduce a counter $\hat{v}$ and enforce the
invariant $\f(v) +\f(\hat{v})=N$ to hold in all the configurations of any run of our counter programs. We achieve this by ensuring that every
line containing an instruction of type $v\pluseq c$ must be followed by a line with a $\hat{v}\minuseq c$ instruction. From
now on, we make the convention that the instruction $v\pluseq c$ is
an abbreviation for $v\pluseq c;\hat{v}\minuseq c$. This
allows us to remove the hat variables from our future counter programs
whenever it is convenient for us, which will ease readability. So, if we choose an initial configuration in which $\f(v)+\f(\hat{v})=N$, we have that this invariant holds 
whenever the zero-test gadget is invoked.  

 We introduce auxiliary counters $u_1, u_2$ that
 will be tested for zero only in the final configuration,
 and hence have no hat counterpart. 
 In the following, the instruction $\Zero{v}$ denotes
 the following gadget:
\begin{algorithm}[H]
\caption{$\Zero{v}$}
\begin{algorithmic}[1]
\Loop
\State $v\pluseq1;\hat{v}\minuseq1;u_2\minuseq1$
\EndLoop
\Loop
\State $v\minuseq1;\hat{v}\pluseq1;u_2\minuseq1$
\EndLoop
\State $u_1\minuseq2$
\end{algorithmic}
\end{algorithm}
Consider an initial configuration in which $\f_1(u_1)=2n$ and $\f_1(u_2)=2n\cdot N$ for some $n>0$. Initially, it is true that $\f(u_2)=
\f(u_1)\cdot N$. 
\begin{lemma}[\cite{czerwinski2020reachability}]
  There exists a run of the counter program $\Zero{v}$ that starts in a configuration with $\f(u_2)\ge 2$ , $\f(u_2)=\f(u_1)\cdot N$, and ends in a configuration with $\f(u_2)=\f(u_1)\cdot N$ if and only if
  $\f(v) = 0$ in the initial configuration.
\end{lemma}
\begin{proof}
  Because of the
  invariant $\f(v)+\f(\hat{v})=N$, the loop on Line~1 can
  decrease the value of $u_2$ by at most $N$ and the loop on line
  3 can also decrease the value of $u_2$ by at most $N$. Moreover,
  this can only happen if $val_1(v)=0$. 
\end{proof}
From a configuration with $\f(u_2)=\f(u_1)\cdot N$, a run ``incorrectly'' executing the  $\Zero{v}$ routine can only reach a configuration with $\f(u_2)>\f(u_1)\cdot N$. Observe that from such a configuration, we can never reach a configuration respecting the invariant $\f(u_2)=\f(u_1)\cdot N$ if we only use $\Zero{v}$ instructions to change the values of $u_1,u_2$. Now, consider a counter $v$ and a counter program $C$ that modifies
the values of counters $u_1$ and $u_2$ only through the $\Zero{v}$
instruction. If we start in a configuration in which $\f(u_1)=2n$ and $\f(u_2)=2n\cdot N$ for some $n>0$, and we are guaranteed that any run of $C$
cannot execute more than $n$ $\Zero{v}$ instructions, then after any
run of $C$, we have that $\f(u_2)=\f(u_1)\cdot N$ only if the value of
the counter $v$ was zero at the beginning of every $\Zero{v}$
instruction. If all the variables that we are interested in are bounded by the same value $N$, we can use a single pair of counters $u_1,u_2$ to perform zero tests on all our variables. We subsequently call the counters $u_1$ and
$u_2$ \emph{testing counters}. To summarise, in VASS, we can perform $n$ zero tests on 
counters bounded by $N$ via reachability queries.

Given a configuration $(L,\f)$, we say that $(L,\f)$ is a \emph{valid
  configuration} if $\f$ respects the condition that $\f(u_2) = \f(u_1)\cdot N$. A \emph{valid run} is a
run that starts in a valid configuration and ends in a valid
configuration. Also, a counter program \emph{admits} a valid run if
there exists a valid run that reaches the terminal instruction \textbf{halt}. 

Having introduced all relevant definitions, we now
introduce components, which are counter programs
acting as sub-routines that ensure that, if
invoked in a configuration fulfilling the invariants
required for zero tests, upon returning those
invariants still hold. Formally, a \emph{component} is a counter program such that:
\begin{itemize}
\item there is a polynomial $p$ such that every valid run performs at most $p(N)$ zero tests on all variables; and
\item the values of the testing counters are updated only by $\Zero{}$
  instructions.
\end{itemize}
We
conclude this section with 
\Cref{lemma:component composition}, which states that sequential
composition and non-deterministic branching of components yields
components. We will subsequently implicitly make use of this
lemma without referring to it.
\begin{lemma}\label{lemma:component composition}
  If $C_1,C_2$ are components then both $C_1;C_2$ and $C_1\textbf{ or }C_2$ are also components.
\end{lemma}

\section{Rudimentary arithmetic and unary VASS}\label{sec-unaryVASS}
In this section, we provide a lower bound for the zero-reachability problem for a VASS when the input configuration is encoded in unary. We observe that there is a close relationship between this problem and deciding
validity of a formula of first-order arithmetic with addition and multiplication on an initial segment of $\N$, also known as
rudimentary arithmetic~\cite{EM98}. 

\subsection{Rudimentary arithmetic}
For the remainder of this section, all the structures we consider are relational.
We denote by $\textbf{FO}(+,\times)$ the first-order theory of the structure $\langle\N,+,\times\rangle$, where $+$ and $\times$ are the natural ternary addition and multiplication relations. When interpreted over initial segments of $\N$, 
the family of the first-order theories is known rudimentary
arithmetic. In order to have a meaningful reduction
to \emph{fixed} VASS, we are interested in the following
decision problem:
\begin{problem}\label{prob:unary fixed FO}
  \textsc{Fixed $\textbf{FO}(\underline N, +,\times)$ Validity} \\
    \textbf{Fixed:} $\Phi(\vect x)\in \textbf{FO}(+,\times)$.\\
  \textbf{Input:} $N\in\N$ and $\vect{x}\in \underline{N}^n$ given in unary.\\
  \textbf{Output:} YES if and only if $\langle \underline{N},+,\times\rangle\models \Phi(\vect x)$.
\end{problem}
Note that, in particular, for a predicate $x + y = z$ to
hold, all of $x, y, z$ must be at most $N$. It thus
might seem that after we fix $N$, a formula $\Phi(\vect x)$ can only express facts about numbers up to $N$. However,
as discussed in~\cite{schweikardt2005arithmetic}
and~\cite{EM98}, this can be improved to quantifying
over variables up to $N^d$ for any fixed~$d$ using
$(N+1)$-ary representations of numbers. In other words, for any fixed $d$ and formula $\Phi(\vect x)$, there exists a formula $\Phi'(\vect x)$ such that for any $N\in\N$ and $\vect x\in \underline{N}^n$, $\langle \underline{N},+,\times\rangle\models \Phi(\vect x)$ iff $\langle \underline{N}^d,+,\times\rangle\models \Phi'(\vect x)$.
\subsection{Reductions between unary languages}
In order to study decision problems whose input is, for some constant $k$, a $k$-tuple of numbers presented in unary, and hence to analyse languages corresponding to them, we need 
a notion of reductions that are weaker compared to the 
standard ones that are widely used in computational complexity. The reason is that classical problems 
involving numbers represented in unary, such as \textsc{Unary Subset Sum}~\cite{EJT10}, have as an input a sequence
of variable length of numbers given in unary. Hence, languages of such problems are in fact binary, as we need a delimiter symbol to separate the elements of the sequence. It is not clear how a reasonable reduction from such a language to a language consisting of
$k$-tuples of numbers for a \emph{fixed} $k$ would look like. Conversely, arithmetic properties of a single number, e.g.\ primality or square-freeness, require very low
computational resources if the input is represented in unary. Hence, the notion of a reduction between such ``genuinely unary'' languages has to be very weak.

In view of this discussion, we introduce the following kind
of reduction. 
Given $k>0$, a \emph{$k$-tuple unary language} is a subset $L\subseteq \N^k$. 
We say that $L$ is a tuple unary language if $L$ is a $k$-tuple unary
language for some $k>0$. Let $L \subseteq \N^k$ and $M \subseteq \N^\ell$ be tuple 
unary languages, we say that $L$ \emph{arithmetically reduces} to $M$ if there are fixed polynomials $p_1,\ldots,p_\ell\colon \N^k \to \N$
such that $(m_1,\ldots,m_k)\in L$ if and only if $(p_1(m_1,\ldots,m_k),\ldots,p_\ell(m_1,\ldots,
m_k)) \in M$.

We believe that this reduction is sensible for the
following informal reasons. Polynomials can be
represented as arithmetic circuits. To the
best of our knowledge, there are no
known lower bounds for, e.g. comparing the
output of two arithmetic circuits with input gates
with constant value one~\cite{ABKM09}, suggesting that
evaluating a polynomial is a computationally weak
operation. Moreover, in the light of sets of
numbers definable in rudimentary arithmetic,
it seems implausible that applying a polynomial
transformation makes, e.g. deciding primality of 
a number substantially easier.

For a formula $\Phi$, let $\mathcal{L}_{\Phi}$ be the tuple unary language of yes-instances for \textsc{Fixed $\textbf{FO}(\underline N, +,\times)$ Validity}. Also, for a counter program $C$, define $\mathcal{L}_{C}$ as the tuple unary language of yes-instance for the \textsc{Fixed counter program zero-reachability} problem. The remainder
of this section is devoted to proving the following
theorem.
\begin{theorem}\label{thm:main-reduction-unary}
For every formula $\Phi$ of $\textbf{FO}(\underline N,+,\times)$, there exists a counter program~$C$ such that $\mathcal{L}_{\Phi}$ arithmetically reduces to $\mathcal{L}_{C}$. 
\end{theorem}
This theorem can be viewed in two different contexts. On the one hand, it relates the computational complexity of the two problems using a very weak reduction as described above. On the other hand, it also relates the expressivity of two formalisms. Namely, the set of satisfying assignments for formulas of rudimentary arithmetic is at most as expressive as the composition of polynomial transformations with the sets of initial configurations for zero-reachable runs in counter programs. In particular, it shows that fixed VASS
can, up to a polynomial transformation, decide
number-theoretic properties such as primality, 
square-freeness, see~\cite{EM98} for further examples.
We now proceed showing how fixed VASS can simulate
formulas of $\textbf{FO}(\underline N,+,\times)$.
\subsection{Components for arithmetic operations}
\label{subsec-arithm-gadgets}
Since there is no straightforward way to model negation with a counter program, we need to provide gadgets for both the predicates $+$ and $\times$ of $\textbf{FO}(\underline N,+,\times)$ and their negation and hence design one separate component for each literal. However, these components may change the values of the counters representing first-order variables, and since a first-order variable might appear in multiple literals, we first provide a gadget to copy the value of a designated counter to some auxiliary counter before it can be manipulated.

\paragraph{Copy.}
We provide a counter program $\textsc{Copy}[x,x']$
that will be invoked several times with the following properties:
\begin{enumerate}
    \item $\textsc{Copy}[x,x']$ admits a valid run if and only if $val_{end}(x')=val_{end}(x)=val_1(x)$; and 
    \item $\textsc{Copy}[x,x']$ is a component. 
\end{enumerate}

\begin{algorithm}[H]
\caption{$\textsc{Copy}[x,x']$}\label{alg:zero-test}
\begin{algorithmic}[1]
    \Loop 
    \State $x'\minuseq1$
    \EndLoop
    \State \Zero{$x'$}
    \ccomment{Reset auxiliary variable $x'$}
    \Loop
    \State $x\minuseq1;x'\pluseq1;t\pluseq1$
    \EndLoop
    \State \Zero{$x$}
    \Loop
    \State $t\minuseq1;x\pluseq1$
    \EndLoop
    \State \Zero{$t$}
\end{algorithmic}
\end{algorithm}
The purpose of the loop on Line~1 is to ensure that $val_4(x')=0$. We do not need to do this for the auxiliary counter $t$ because any valid run sets $val_{end}(t)=0$. Next, observe that $\textsc{Copy}[x,x']$ admits a valid run if and only if the loop on Line~4 is executed $val_1(x)$ many times and the loop on Line~7 is executed $val_4(t)=val_1(x)$ many times which happens if and only if $val_{end}(x')=val_{end}(x)=val_1(x)$. Moreover, any valid run performs exactly 3 zero tests, so $\textsc{Copy}[x,x']$ is a component. The auxiliary variable $t$ is reused between different $\textsc{Copy}$ components, so that we do not introduce new counters every time we use the $\textsc{Copy}$ component.

\paragraph{Addition.}
We define a counter program $\textsc{Addition}[x,y,z]$ that enables us to check whether the value stored in counter $z$ is equal to the sum of the values stored in $x,y$: 
\begin{enumerate}
\item $\textsc{Addition}[x,y,z]$ admits a valid run if and only if $val_1(x)+val_1(y) = val_1(z)$;
\item $\textsc{Addition}[x,y,z]$ is a component; and 
\item the effect of $\textsc{Addition}[x,y,z]$ is zero on counters $x,y,z$.
\end{enumerate}
\begin{algorithm}[H]
\caption{$\textsc{Addition}[x,y,z]$}
\begin{algorithmic}[1]
\State $\textsc{Copy}[x,x'];\textsc{Copy}[y,y'];\textsc{Copy}[z,z']$
\Loop
\State $z'\minuseq1$
\State $x'\minuseq1\textbf{ or }y'\minuseq1$
\EndLoop
\State $\textbf{zero-test}(x');$ $\textbf{zero-test}(y');$
$\textbf{zero-test}(z')$
\end{algorithmic}
\end{algorithm}

It is easy to see that the first property is fulfilled by the counter program and 
that $\textsc{Addition}[x,y,z]$ is a component because any run performs exactly 12 zero tests (9 tests on Line~1, and 3 tests on Line~5). The last property is true based on the properties of $\textsc{Copy}$.

The component for $\neg\textsc{Addition}[x,y,z]$ that checks whether the value stored in counter $z$ is different from the sum of the values stored in counters $x,y$ is defined similarly, see \Cref{app:not addition} for details. 

\paragraph{Multiplication.}
We now define a counter program $\textsc{Multiplication}[x,y,z]$ with the following properties:
\begin{enumerate}
\item it admits a valid run if and only if $val_1(z)=val_1(x)\cdot val_1(y)$;
\item $\textsc{Multiplication}[x,y,z]$ is a component; and 
\item the effect of $\textsc{Multiplication}[x,y,z]$ is zero on counters $x,y,z$.
\end{enumerate}
\begin{algorithm}[H]
\caption{$\textsc{Multiplication}[x,y,z]$}
\begin{algorithmic}[1]
\State $\textsc{Copy}[x,x'];\textsc{Copy}[y,y'];\textsc{Copy}[z,z']$
\Loop
\Loop
\State $x'\minuseq1;t\pluseq1;z'\minuseq1$
\EndLoop
\State $\Zero{x'}$
\Loop
\State $x'\pluseq1;t\minuseq1;$
\EndLoop
\State $\Zero t$
\State $y'\minuseq1$
\EndLoop
\State $\Zero{y'};\Zero{z'}$
\end{algorithmic}
\end{algorithm}

Observe that the loop on Line~3 of any valid run must be executed $val_1(x)$ many times in order to pass the zero test on Line~5. The effect of this loop is then to decrease the value of $z'$ by $val_1(x)$ and to set the value of $t$ to $val_1(x)$. Next, the loop on Line~6 must be executed $val_5(t)=val_1(x)$ many times to pass the zero test on Line~8, so the value of $x'$ is set to $val_1(x)$ and the value of $t$ is set again to zero. Hence, the effect of Lines~3-8 is to subtract $val_1(x)$ from the value of $z'$ without changing the value of $x'$. Finally, any valid run passes the test on Line~10 if and only if the loop on Line~2 is executed $val_1(y)$ many times, which happens if and only if $val_1(z) = val_1(x)\cdot val_1(y)$. Since, we argued that the loop on Line~2 is executed $val_1(y)$ many times, we conclude that any valid run of $\textsc{Multiplication}[x,y,z]$ performs at most $2N+9$ zero tests, so $\textsc{Multiplication}[x,y,z]$ is a component. Again, the last property is ensured by the properties of $\textsc{Copy}$.
The definition of 
$\neg\textsc{Multiplication}[x,y,z]$ is similar, see \Cref{app:not multiplication} for futher details.

\subsection{Components for quantification}

We define the last components that we need in order to prove \Cref{thm:main-reduction-unary}. These components allow us to existentially and
universally quantify over variables in a bounded range.

We begin with the existential quantifier and we design a 
counter program
$\textsc{Exists}[v]$ with the following properties:
\begin{enumerate}
    \item for every $n\in\underline{N}$, $\textsc{Exists}[v]$ admits a valid run
      $\varrho$ such that $\val_{end}(\varrho,v)=n$;
    \item $\textsc{Exists}[v]$ is a component. 
\end{enumerate}
We define $\textsc{Exists}[v]$ as follows:
\begin{algorithmic}[1]
\Loop
\State $v\pluseq1$
\EndLoop
\State $\textbf{skip}$
\end{algorithmic}
It is easy to see that both properties hold, since $v$
cannot attain a value larger than $N$.

While the component used for simulating existential 
quantification can be sequentially composed with
a component for a subformula, universal quantification
requires directly integrating the component over
whose variable we universally quantify.
Let $C[v]$ be a component that may access the counter
$v$, test it for zero, and change its value on intermediate steps, but
has overall net effect zero on counter $v$. We write
$\textsc{ForAll}[v]:C[v]$ 
for the following counter program:
\begin{algorithmic}[1]
\Loop
\State $C[v]$
\State $v\pluseq1$
\EndLoop
\State \Zero{$\hat{v}$}
\end{algorithmic}
The properties of $\textsc{ForAll}[v]:C[v]$ are as follows:
\begin{enumerate}
\item it admits a valid run if and only if for all
  $u\in\underline{N}$, $C$ has a valid run with $\val_1(v)=u$; and
\item $\textsc{ForAll}[v]:C[v]$ is a component.
\end{enumerate}

Notice that the instruction on Line~4 tests if $val_4(v)=N$. Thus, any valid run that passes the test on Line~4 must be able to execute $C[v]$ for all values of $v\in\underline{N}$. Moreover, since $C[v]$ is a component, we know that it executes a bounded number of zero test $B$, so $\textsc{ForAll}[v]:C[v]$ executes at most $N\cdot B$ many zero tests and it is thus a component.

\subsection{Putting it all together}
Having defined all the building blocks above, we now
prove \Cref{thm:main-reduction-unary}, which is a result of the following lemma.
\begin{lemma}\label{prop:main-reduction-unary}
  For any formula $\Phi(\vect x)$ of $\textbf{FO}(\underline N, +,\times)$, there exists a component $C$ and polynomials $p_1,\ldots,p_k:\N^{(n+1)}\rightarrow \N$ such that for any $N\in\N$ and $\vect x\in\N^n$, $\langle\underline{N},+,\times\rangle\models \Phi(\vect x)$ if and only if $C$ admits a valid run from the initial configuration $(p_1(N,\vect x),\ldots,p_k(N,\vect x))$.
\end{lemma}
\begin{proof}
    Assume with no loss of generality that $\Phi(\vect x)$ is in prenex form, i.e. $\Phi(\vect x)\equiv\exists \vect y_1\forall \vect y_2\ldots Q_m\vect y_m\varphi(\vect x,\vect y_1,\ldots,\vect y_m)$ and that the negations are pushed to the predicate level. Let $M$ be the number of quantified variables and $K$ be the number of literals in $\Phi(\vect x)$. Also, let $\vect a$ be the constants of $\Phi(\vect x)$.
The counter variables of the component $C$ are 
\begin{itemize}
    \item $\vect x$ representing the free variables of $\Phi(\vect x)$,
    \item $\vect y$ representing the quantified variables of $\Phi(\vect x)$,
    \item $\vect c$ representing the constants of $\Phi(\vect x)$; and 
    \item the relevant auxiliary variables $t,\vect x',\vect y',\vect c'$ used inside the components for $\textsc{Copy}[x,x']$, $\textsc{Multiplication}[x,y,z]$ and $\neg\textsc{Multiplication}[x,y,z]$.
\end{itemize}
such that
\begin{itemize}
\item $\f_1(x)=x$ and $\f_1(\hat{x})=N-x$ for all
  variables $x$ occurring in $\vect x$;
\item $\f_1(v)=0$ and $\f_1(\hat{v})=N$ for all the counters representing quantified variables, constants and auxiliary counters; and
\item for the testing counters, $\f_1(u_1)=2N$ and $\f_1(u_2)=2N\cdot (K\cdot(2N+12))^M$.
\end{itemize}
The counter program $C$ starts with a component $C_0$ that initialises the constant variables $\vect c$ of $\Phi(\vect x)$ by a sequence of instruction of the type $c\pluseq a$, for every constant appearing in $\Phi(\vect x)$. 

If $\varphi(\vect x, \vect y_1,\ldots,\vect y_m)$ contains a single literal, we can construct a fixed component $C'$ such that the literal is true under the induced assignment by a valid configuration $(L,\f)$ if and only if $C'$ admits a valid run starting in $(L,\f)$ by using the previously defined components. Thus, we can construct a component $C_\varphi$ that
admits a valid run starting from some valid configuration $(L,\f)$ if
and only if $\varphi$ evaluates to true under the assignment induced by
$(L,\f)$, by simulating a conjunction via sequential composition and
disjunction by non-deterministic branching.

We now turn towards simulating the quantifiers in $\Phi(\vect x)$. For $1\le
i\le m$ and $\vect y_i=(y_1,\ldots,y_{m_i})$, the component $C_i$ is
obtained as follows:

\vspace*{0.2cm}
\begin{minipage}[h]{0.45\linewidth}%
\begin{itemize}
\item If $Q_i=\exists$, $C_i$ is the component
  \begin{algorithmic}[1]
    \State $\textsc{Exists}[y_1]$
    \State $\textsc{Exists}[y_2]$
    
    \break
    $\vdots$
    \State $\textsc{Exists}[y_{m_i}]$
    \State $C_{i+1}$
  \end{algorithmic}
\end{itemize}
\end{minipage}
\begin{minipage}[h]{0.5\linewidth}
\begin{itemize}
\item If $Q_i=\forall$ then $C_i$ is the component
  \begin{algorithmic}[1]
    \State $\textsc{ForAll}[y_1]:$
    \State \hspace{\algorithmicindent}$\textsc{ForAll}[y_2]:$
    
    \break
    \hspace{\algorithmicindent}$\vdots$
    \State \hspace{\algorithmicindent*2}$\textsc{ForAll}[y_{m_i}]:$
    \State \hspace{\algorithmicindent*3}$C_{i+1}$
  \end{algorithmic}
\end{itemize}
\end{minipage}
\vspace*{0.2cm}

Finally, we let $C=C_0;C_1$ and $C_{m+1}=C_\varphi$. By the properties
established above, it is clear that $C$ admits a
valid run starting with $\f_1$ defined above if and only if $\Phi(\vect x)$ is
valid. Indeed, the sequential composition of the components $C_1$ to
$C_m$ correctly implements the range-bounded alternating
quantification. Every time
$C_\varphi$ is executed, due to $C_\varphi$ being a component it is
guaranteed that it is traversed by a valid run if and only if $\varphi(\vect x,\vect y_1,\ldots,\vect y_m)$
evaluates to true under the assignment encoded in the configuration
when entering $C_\varphi$.

To see that $C$ is a component, it remains to be shown that any valid run of
$C$ cannot perform more than $(K\cdot(2N+12))^M$ zero tests. By the properties of our components, we know that any valid run in $C_{\varphi}$ performs at
most $K\cdot(2N+12)$ zero tests on any
variable. Furthermore, if any valid run of $C_{i+1}$ performs at most
$B$ zero tests for any variable then any valid run in $C_i$ performs
at most the following number of zero tests:
\begin{itemize}
\item $B$ if $Q_i=\exists$; and
\item $N^{m_i} B$ if $Q_i=\forall$.
\end{itemize}
Thus, any valid run of $C$ performs at most $(K\cdot(N+12))^M$
zero tests on any variable. Finally, observe that $C$ itself is fixed by the structure of
$\Phi(\vect x)$. In particular, the only variable part, the values of the free variables $\vect x$ and the bound $N$, are fully encoded in the initial configuration $\f_1$.
\end{proof}

To prove \Cref{thm:main-reduction-unary}, add a loop repeating zero tests at the end of $C$, thus setting the values of the testing counters to zero if and only if the invariant described in \Cref{sec-gadgets} holds. After that, set to zero all the remaining counters (including the hatted ones) by decrementing them in loops. A run in thus constructed counter program is zero-accepting if and only if it is valid.

\section{A universal VASS for polynomial space computations}\label{sec-binaryVASS-new}
The goal of this section is to show that there is 
a fixed 5-VASS whose zero-reachability problem
is PSPACE-hard, provided that the initial
configuration is encoded in binary. Let us first
remark that  we can actually use the techniques developed in
the previous section to prove that for every $i$, there exists a fixed VASS $V_i$ such that deciding zero-reachability for $V_i$ is $\Sigma_i^{\textrm P}$-hard. A result
by Nguyen and Pak~\cite{NguyenP22} shows that for
every $i$, there is a fixed formula of Presburger
arithmetic $\Phi_i$, i.e. $\mathbf{FO}(\N,+)$, such
that deciding $\Phi_i$ is $\Sigma_i^{\textrm P}$-hard.
It is well-known that quantification in Presburger
arithmetic relativises to an initial segment $\underline N$
for some $N\in N$ whose bit length is polynomial in the 
size of $\Phi_i$, which is fixed, see e.g.~\cite{weispfenning1990complexity}. Hence, by
combining the results from~\cite{NguyenP22} with
\Cref{prop:main-reduction-unary}, it is possible
to show that zero-reachability for fixed binary VASS 
is hard for the polynomial hierarchy.
We do not explore this method further because we can actually construct a fixed binary VASS such that the zero-reachability problem is PSPACE-hard for it and which has a smaller number of counters than the fixed binary VASS obtained from showing NP-hardness via the reduction from short Presburger arithmetic outlined above.

We proceed with our construction as follows. We start with the halting problem for Turing machines (TMs) working in polynomial space and show that this problem is PSPACE-hard even if the space complexity of the TM is bounded by the length of its encoding and its input is empty. In \Cref{prop:universal-tm}, we then reformulate the halting problem as follows: given the encoding of such a machine as an input to a universal one-tape TM $\UU$, does $\UU$ accept?

We then use two consecutive simulation. First, we simulate $\UU$ with a $3$-counter automaton $\Aa$ (\Cref{prop:sim-tm-ca}), and then simulate $\Aa$ with an $5$-VASS $\VV$ (\Cref{thm:binary-main}). To be able to apply the technique described in \Cref{sec-gadgets}, we make sure that the space complexity stays polynomial in the size of the input throughout these simulations. This implies that both the upper bound on the value of the counters and the required number of zero tests are polynomial in the size of the input, which enables us to establish a polynomial time reduction.
As a result we obtain a VASS $\VV$ which, in a certain sense, can simulate arbitrary polynomial-space computations.

To provide the reduction, we then show how to transform in polynomial time the input of the problem we started with, the halting problem for polynomial-space TMs, into a zero-reachability query for $\VV$.

\subsection{The halting problem for space-bounded TMs}

The goal of this subsection is to show that there exists a \emph{fixed} polynomial-space TM whose halting problem is PSPACE-complete.

\begin{proposition}[{\cite[Section 4.2]{Arora2009}}]
    The following problem is PSPACE-complete: given a TM~$\MM$, an input word $w$ and a number $n$ encoded in unary, decide if $\MM$ accepts $w$ in at most $n$ space.
\end{proposition}

Note that, using standard arguments, we can without loss of generality assume that $\MM$ always halts. 

We fix some way of encoding, using an alphabet of size at least two, of Turing machines and we denote by $\norm{\MM}$ the length of the encoding of $\MM$, which we call the \emph{size} of $\MM$. Given a TM $\MM$, we say that it is $\norm{\MM}$-space-bounded if on every input it halts using at most $\norm{\MM}$ space. Given $\MM$, an input word $w$ and a number $n$ encoded in unary, it is easy to construct a $\norm{\MM}$-space-bounded TM $\MM'$ such that if $\MM$ accepts $w$ in space at most $n$, then $\MM'$ accepts on the empty input, otherwise $\MM'$ rejects on the empty input. Moreover, the size of $\MM'$ is polynomial in $\norm{\MM}$, $|w|$ and $n$.

Indeed, $\MM'$ can be constructed as follows. When run on the empty input, it writes $w$ on some tape, and then runs $\MM$ treating this tape as the input tape. Additionally, it initialises another tape with $n$ written in unary, and before each step of $\MM$ it checks that the space used by the tape where $\MM$ is simulated does not exceed $n$. If it does, it immediately rejects. It is easy to see that such a TM is $\norm{\MM'}$-space-bounded and satisfies the required conditions.

Hence we get that the following problem is PSPACE-complete: given a $\norm{\MM}$-space-bounded TM $\MM$, does $\MM$ accept on the empty input? Observe that from the construction above we can assume that $\MM$ has a special representation such that the fact that it is $\norm{\MM}$-space-bounded can be checked in polynomial time. 

Let $\UU$ be a one-tape universal TM. This TM has a single read-write tape, which in the beginning contains the input, that is, a description of a TM $\MM$ it is going to simulate. If $\MM$ is $\norm{\MM}$-space-bounded (and represented as mentioned in the previous paragraph), $\UU$ simulates $\MM$ on the empty input in space polynomial in $\norm{\MM}$ \cite[Claim 1.6]{Arora2009}, otherwise $\UU$ rejects. That is, in this space, $\UU$ accepts or rejects depending on whether $\MM$ accepts or rejects the empty word. Hence we get the following proposition.

\begin{proposition}\label{prop:universal-tm}
    There exists a fixed polynomial-space TM $\UU$ such that the question whether $\UU$ halts on a given input is PSPACE-complete.
\end{proposition}

\subsection{From TMs to a counter automata}
\label{subsec-tm-to-ca}

In previous subsection, we obtained a PSPACE-complete problem which already resembles the form of the reachability problem for a fixed counter program: given a fixed polynomial-space TM $\UU$ does it accept a given input? In this section we show how to simulate $\UU$ with a fixed counter automaton $\Aa$ (and how to transform its input accordingly), and in the next section we show how to simulate $\Aa$ with a fixed binary
VASS $\VV$.

Let $\Aa$ be a counter automaton. We say that $\Aa$ is \emph{deterministic} if for every configuration $(q, n_1, \ldots, n_d)$ there is at most one transition that $\Aa$ can take from this configuration. Suppose that $\Aa$ is deterministic, and that its final state $q_f$ does not have any outgoing transitions. Let $\vect{n} = (n_1, \ldots, n_d) \in \NN^d$. We treat~$\Aa$ as an acceptor for such vectors. We say that $\Aa$ works in time $t$ and space $s$ on $\vect{n}$ if the unique run starting in the configuration $(q_0, n_1, \ldots, n_d)$ ends in a state without outgoing transitions, has length $t$, and the bit length of the largest value of a counter along this run is $s$. If this run ends in $q_f$, we say that $\Aa$ \emph{accepts} this vector, otherwise we say that it \emph{rejects} it. In all our constructions we make sure that there are no infinite runs.
Note that, as in the case of TMs, we measure space complexity in the bit length of the values of the counters, and not in their actual values.

Let $\Sigma$ be a finite alphabet. Let us bijectively assign a natural number to each word over~$\Sigma$ as follows. First, assign a natural number between $1$ and $|\Sigma|$ to each symbol in~$\Sigma$. Then $w$ can be considered as a number in base $|\Sigma| + 1$, with the least significant digit corresponding to the first letter of $w$. We denote this number by $\num(w)$.

Let $\MM$ be a TM, and $w$ be its input. We can transform $w$ into a vector $(\num(w), 0, \ldots, 0)$, which will constitute an input of a deterministic counter automaton $\Aa$. We say that $\Aa$ \emph{simulates} $\MM$ if a word $w$ is accepted by $\MM$ if and only if the corresponding vector is accepted by $\Aa$. Furthermore, we say that this simulation is in polynomial space if there exists a polynomial $p$ such that on every input, the space complexity of $\Aa$ is at most $p(S)$, where $S$ is the space complexity of $\MM$ on this input.

The proof of the following proposition uses the techniques described in the proofs of \cite[Theorem~4.3(a)]{Fischer1968} and \cite[Theorem~2.4]{Greibach1976}.

\begin{proposition}
\label{prop:sim-tm-ca}

For every one-tape TM $\MM$, there exists a deterministic $3$-counter automaton $\Aa$ that simulates it in polynomial space.
\end{proposition}

\begin{proof} The idea of the proof is as follows. Two counters of $\Aa$, call them $\ell$ and $r$, represent the content of the tape of $\MM$ to the left and to the right of the reading head. They are encoded similarly to the way we encode the input word. Namely, let $w_1aw_2$, where $w_1, w_2 \in \Sigma^*$ and $a \in \Sigma$ be the content of the tape at some moment of time, with the working head in the position of the letter $a$. Denote by $w_1^R$ the reversal of the word $w$. Then $\ell$ stores $\num(w_1^R)$, $r$ stores $\num(w_2)$, and $a$ is stored in the finite memory of the underlying finite automaton. 

Now, to make a step to the left, we do the following. First, we need to add $a$ to the end of the word encoded by the value of $r$. This is done by multiplying the value of $r$ by $|\Sigma| + 1$ and adding $\num(a)$ to it. Next, we need to extract the last letter of the word encoded by the value of $\ell$, and remove this letter. To do so, we do the opposite of what we did for $r$: this letter is the residue of dividing the value of $\ell$ by $|\Sigma| + 1$, and the new value of $\ell$ is the result of this division. 

The reason we need the third counter $x$ is to perform these multiplications and divisions. Namely, to divide the value of a counter $\ell$ by a constant $c$, we repeat the following until it is no longer possible: subtract $c$ from the value of $\ell$ and add one to the value of $x$. When the value of $\ell$ becomes smaller than $c$, we get the result of the division in the counter $x$, and the remainder in $\ell$. Multiplication by a constant is done similarly.

Observe that by construction the largest value of a counter of $\Aa$ at any moment of time is at most $(|\Sigma| + 1)^S$, where $S$ is the maximal amount of space $\MM$ uses on given input. Hence $\Aa$ works in space polynomial in the length of the input. 
\end{proof}

By simulating $\UU$ from   \Cref{prop:universal-tm} with a counter automaton $\Aa$, we get the following statement.

\begin{corollary}\label{corr:fixed-ca} There exists a fixed $3$-counter automaton $\Aa$ working in polynomial space (in the length of the input) such that the zero-reachability problem for it is PSPACE-complete.
\end{corollary}

The same question can be asked about fixed $2$-counter automata. Informally speaking, such automata are exponentially slower than $3$-counter automata: the known simulation requires
storing the values of the three counters $x, y, z$ as $2^x 3^y 5^z$ \cite{Minsky1967}. They are also less expressive: for example, 2-counter automata cannot compute the function $2^n$ \cite{Schroeppel1972}, while for $3$-counter automata this is trivial. It is worth noting the developments of the next
subsection imply that a lower bound for fixed 2-counter
automata translates into a lower bound for fixed 4-VASS.

\subsection{From counter automata to VASS}
\label{subsec:ca-to-cp}

To go from a counter automaton to a VASS, we need to simulate zero tests with a VASS. In general, this is not possible. However, the counter automaton in Corollary \ref{corr:fixed-ca} is space bounded, so we can deduce from it PSPACE-hardness of the zero-reachability problem in fixed $8$-VASS using the technique described in Section \ref{sec-gadgets}. However, using a more advanced technique of quadratic pairs described in \cite{Czerwinski2022} allows us to deduce the same result for $5$-VASS. Namely, a slight variation of \cite[Lemma 2.7]{Czerwinski2022} states that given a $3$-counter automaton $\Aa$ working in polynomial space, one can construct a $5$-VASS $\VV$ such that fixed zero-reachability in $\Aa$ can be reduced in polynomial time to fixed zero-reachability in~$\VV$. Indeed, the fact that $\Aa$ works in polynomial space allows us to initialise the counters of $\VV$ to account for enough zero tests, since it implies that if there exists an accepting run in $\Aa$, then there is a such a run with polynomial number of zero tests and polynomial upper bound on the sum of the values of the counters along the whole run. 

%The technique described in Section \ref{sec-gadgets} thus allows us to deduce PSPACE-hardness of the zero-reachability problem in fixed $8$-VASS from \Cref{corr:fixed-ca} if the input configuration is given in binary. However, a more advanced technique of quadratic pairs 

\begin{theorem}\label{thm:binary-main}
    There exists a fixed $5$-VASS such that the zero-reachability problem for it is PSPACE-hard if the input
    configuration is given in binary. 
\end{theorem}

\bibliographystyle{alpha}
\bibliography{biblio}

\appendix
\section{More arithmetical components}

\subsection{Component for the negation of the addition predicate}\label{app:not addition}

We define the $\neg\textsc{Addition}[x,y,z]$ component that checks if the value stored in counter $z$ is different to the sum of the values stored in the counters $x,y$. Formally, $\neg\textsc{Addition}[x,y,z]$
\begin{enumerate}
\item admits a valid run if and only if $val_1(z)\neq val_1(x)+ val_1(y)$;
\item $\neg\textsc{Addition}[x,y,z]$ is a component; and
\item the effect of $\neg\textsc{Addition}[x,y,z]$ is zero on counters $x,y,z$.
\end{enumerate}
The component $\neg\textsc{Addition}[x,y,z]$ is now defined as follows:
\begin{algorithm}[H]
\caption{$\neg\textsc{Addition}[x,y,z]$}
\begin{algorithmic}[1]
\State $\textsc{Copy}[x,x'];\textsc{Copy}[y,y'];\textsc{Copy}[z,z']$
\Loop
\State $z'\minuseq1$
\State $x'\minuseq1\textbf{ or }y'\minuseq1$
\EndLoop
\State $\Gotoo{6}{9}$
\ccomment{Check if $z'=0$ and at least one of $x',y'$ are greater than zero}
\State $\Zero{z'}$
\State $x'\minuseq1\textbf{ or }y'\minuseq1$
\State $\Goto{11}$
\ccomment{Check if $x'=y'=0$ and $z'$ is greater than zero}
\State $\textbf{zero? } x',y'$
\State $z'\minuseq1$
\State $\textbf{skip}$
\end{algorithmic}
\end{algorithm}
The instructions on Lines~2-4 are the same as in the $\textsc{Addition[x,y,z]}$ component, but on Line~5 a valid run must either guess that the value of $z'$ is zero and one of the values of $x'$ or $y'$ are greater than zero, or that the values of $x',y'$ are zero and the value of $z'$ is greater than zero which happens if and only if $val_1(z)\neq val_1(x)+val_1(y)$. Again, $\neg\textsc{Addition[x,y,z]}$ is a component because any valid run performs at most 11 zero tests. The last property is true based on the properties of $\textsc{Copy}$.

\subsection{Component for the negation of the multiplication predicate}\label{app:not multiplication}

The properties of $\neg\textsc{Multiplication}[x,y,z]$ are:
\begin{enumerate}
\item It admits a valid run if and only if $val_1(z)\neq val_1(x)\cdot val_1(y)$;
\item it is a component; and
\item it's effect is zero on counters $x,y,z$.
\end{enumerate}
We define it as the counter program:
\begin{algorithm}[H]
\caption{$\neg\textsc{Multiplication}[x,y,z]$}
\begin{algorithmic}[1]
\State $\textsc{Copy}[x,x'];\textsc{Copy}[y,y'];\textsc{Copy}[z,z']$
\Loop
\Loop
\State $\Gotoo{5}{14}$
\State $x'\minuseq1;t\pluseq1;z'\minuseq1$
\EndLoop
\State $\Zero x'$
\Loop
\State $x'\pluseq1;t\minuseq1;$
\EndLoop
\State $\Zero t$
\State $y'\minuseq1$
\EndLoop
\ccomment{$y'=0$ and $z'>0$}
\State $\textbf{zero? }y'$
\State $z'\minuseq1$
\State $\Goto{16}$
\ccomment{$z'=0$ and $y'>0$}
\State $\textbf{zero? }z'$
\State $y'\minuseq1$
\State $\textbf{skip}$
\end{algorithmic}
\end{algorithm}
Assume that $\neg\textsc{Multiplication}[x,y,z]$ admits a valid run $\pi$. This run must either branch to Line~11 or to Line~14. 
\begin{itemize}
    \item Assume that $\pi$ branches to Line~11. Thus, we can consider that $\pi$ ignores the instruction on Line~4. Observe that Lines~2-10 are exactly the same as in the $\textsc{Multiplication}[x,y,z]$ component apart from Line~4. So, $\pi$ reaches and passes the zero test on Line~11 if and only if $val_1(z)\ge val_1(x)\cdot val_1(y)$, thus, $\pi$ reaches the terminal instruction if and only if $val_1(z)> val_1(x)\cdot val_1(y)$.
    \item Now, assume that $\pi$ branches to Line~14. The instruction on Line~15 can only be executed if the loop on Line~2 was taken at most $val_1(y)-1$ many times. Thus, $\pi$ passes the zero test on Line~14 and executes Line~15 if and only if $val_1(z)< val_1(x)\cdot val_1(y)$. 
\end{itemize} 
Adding the two cases, we obtain that $\neg\textsc{Multiplication}[x,y,z]$ admits a valid run if and only if $val_1(z)\neq val_1(x)\cdot val_1(y)$. Finally, the loop on Line~2 is executed at most $val_1(y)$ many times, so $\neg\textsc{Multiplication}[x,y,z]$ can perform at most $2N+9$ zero tests and thus, it is a component. Finally, the last property is ensured by the properties of $\textsc{Copy}$.

\end{document}